\documentclass[conference]{IEEEtran}

\ifCLASSINFOpdf
\else
\fi
\usepackage{latexsym}
\usepackage{amssymb, amsmath}
\usepackage{graphicx} %package to manage images
\usepackage{amsthm}

\newtheorem{theorem}{Theorem}[section]

\newtheorem{construction}[theorem]{Construction}

\begin{document}

\title{Asymptotically MDS Array BP-XOR Codes}

\author{\IEEEauthorblockN{\c{S}uayb \c{S}. Arslan}
\IEEEauthorblockA{Department of Computer Engineering\\
MEF University\\
Maslak, Istanbul 34099\\
Email: arslans@mef.edu.tr}
%\and
%\IEEEauthorblockN{Homer Simpson}
%\IEEEauthorblockA{Twentieth Century Fox\\
%\and
%\IEEEauthorblockN{James Kirk\\ and Montgomery Scott}
%\IEEEauthorblockA{Starfleet Academy\\
%San Francisco, California 96678--2391\\
%Fax: (888) 555--1212}
}


%Telephone: (800) 555--1212, Fax: (888) 555--1212}


\maketitle

\begin{abstract}
Belief propagation or message passing on binary erasure channels (BEC) is a low complexity decoding algorithm that allows
the recovery of message symbols based on bipartite graph prunning process. Recently, array XOR codes have attracted attention for storage systems due to their burst error recovery performance and easy arithmetic based on Exclusive OR (XOR)-only logic operations. Array BP-XOR codes are a subclass of array XOR codes that can be decoded using BP under BEC. Requiring the capability of BP-decodability in addition to Maximum Distance Separability (MDS) constraint on the code construction process is observed to put an upper bound on the maximum achievable code block length, which leads to the code construction process to become a harder problem. In this study, we introduce asymptotically MDS array BP-XOR codes that are alternative to exact MDS array BP-XOR codes to pave the way for easier code constructions while keeping the decoding complexity low with an asymptotically vanishing coding overhead. We finally provide and analyze a simple code construction method that is based on discrete geometry to fulfill the requirements of the class of asymptotically MDS array BP-XOR codes.
\end{abstract}


% \fi
\IEEEpeerreviewmaketitle

\section{Introduction}
Array codes are linear codes defined for two dimensional data structures that are defined by both data and parity values organized in a matrix form. These codes are quite attractive candidates for burst error recovery in communication and distributed storage systems \cite{Farrell} and provide data reliability with optimal time/space consumption using Maximum Distance Separability (MDS) constraint in the code construction process. Moreover,  a great deal of work has been done and many improvements have been proposed for these codes over the years \cite{Blaum} to secure simpler math and low-complexity computations while still maintain the MDS property.

Typically, any linear code can be represented using a bipartite graph either using the parity check matrix or the generator matrix of the code \cite{shulin}. Using the generator matrix representation, the corresponding bipartite graph has two types of nodes: Nodes that are used to decode (check or coded nodes) and nodes that are decoded (information nodes). Nodes in bipartite graph representation are connected with edges to represent node adjacency. The neighbors of node $j$ (neighbor set), denoted by $\mathcal{N}_j$, is the set of all nodes connected to node $j$. The cardinality of the neighbour set is called the \emph{degree} of node $j$. The Belief Propagation (BP) algorithm a.k.a. message passing algorithm is an iterative process (updating nodes and edges) to decode data from coded nodes over symmetric erasure channels using the bipartite representation of the code. At the onset of the BP algorithm, we begin by setting all the contents of information nodes to NULL that need to be decoded. Then, we look for a degree-one coded node and copy the content to its neighbor information node  by replacing NULL. Next, we update all the coded nodes that are connected to the this neighbor and eliminate the edges that established neighborhood relationship. This completes the first step, and in the next iteration we continue applying the same methodology until there remains no information node with NULL content. If algorithm stops prematurely during iteration, we claim a decoding failure, otherwise we report a decoding success.

Array codes have recently been studied under BP decoding \cite{Wang} and useful upper bounds are derived in \cite{BPXOR} that theoretically establishes the relationship between the block length (and hence the rate of the code), decodability and sparsity of the generator matrix i.e., the encoding/decoding complexity of the code. In this study, we shall demonstrate by relaxing the MDS constraint on the code construction process, we can also dramatically relax the previously found bounds on the code block length \cite{BPXOR} while keeping low complexity BP algorithm successfully decode the whole data block. Such an observation shall yield easier and more powerful code constructions. For instance, we shall consider one of the discrete geometry based codes known as Mojette codes that are recently studied within the context of low density parity check codes and are shown to reduce the node repair complexity  \cite{ArslanMojLDPC}. In our study, we  demonstrate an asymptotically MDS BP-XOR code construction method based on Mojette geometry. By providing and establishing an appropriate set of code parameters, we explicitly construct codes that fulfills the desired theoretical requirements.  

The rest of the paper is organized as follows. In Section \ref{SectionAMDS}, we provide the basics of array MDS BP-XOR codes and give some known results as well as state the main result of the paper. In Section \ref{SectionMoj}, we provide a discrete geometry construction of an asymptotically-MDS array BP-XOR codes. In Section \ref{SectionSim}, we validate our theoretical results by numerically plotting rate, code block length for discrete geometry construction. Finally, we conclude our paper in Section \ref{SectionConc}.

\section{Asymptotically MDS Array BP-XOR Codes}
\label{SectionAMDS}

Before defining the class of asymptotically MDS array BP-XOR codes, let us provide the conventional definition of MDS BP-XOR codes using the notation of reference \cite{BPXOR}.  

\subsection{Background}

Let $l$ be the symbol size in bits and $M = \{0,1\}^l$ be the symbol set from which we select our information as well as coded symbols. The fundamental operation we use is the Exclusive OR (XOR) that is used to add symbols logically bit by bit in binary domain. In our study, nodes represent blocks of data that contains one or more symbols in it. Symbols are the smallest data unit over which XOR operations are defined. 

An $[n,k,t,b]$ array BP-XOR code is a $b \times n$ two dimensional rate $r=k/n$ binary linear code $\mathcal{C} = [a_{i,j}]_{1 \leq i \leq b, 1 \leq j \leq n}$ in which the coding symbol $a_{i,j} \in M$ is the XOR of a subset of source symbols $I = \{v_1,\dots,v_{bk}\}$, typically structured as a $b \times k$ data matrix, and $I$ can be reconstructed from any $n-t$ columns of the linear code $\mathcal{C}$ using BP algorithm for an appropriate integer $t \leq n-k$.  The degree of a coded symbol $a_{i,j}$, denoted as $\sigma_{i,j}$, is the number of information symbols that participate in logical XOR operation i.e., $a_{i,j} = v_{z_1} \oplus \dots \oplus v_{z_{\sigma_{i,j}}}$ such that $v_{z_s} \in I$ for all $s \in \{1,\dots,\sigma_{i,j}\}$. A $t$-erasure correcting array BP-XOR code is MDS if the source symbols can be reconstructed from $k = n-t$ columns of $\mathcal{C}$.

For a given positive integer $b^\prime$ satisfying $b^\prime > b$, a $[n,k,t,b, b^\prime]$ asymptotically MDS array BP-XOR code $\mathcal{C}^a$ is a linear code with $i$-th column $(y_{i,1}, \dots, y_{i,b_i}) = (x_{1}, \dots, x_{bk})G_i$ for a $bk \times b_i$ generator matrix $G_i, i \in \{1,\dots,n\}$ such that $b^\prime = (1/n)\sum_ib_i$. Thus, the generator matrix for $\mathcal{C}^a$ is given by the $bk \times \sum_i b_i$ matrix,
\begin{eqnarray}
G_{\mathcal{C}^a} = [G_1 | G_2 | \dots | G_n].
\end{eqnarray}

What makes this code asymptotically MDS is that it is possible to perfectly reconstruct user data matrix $I$ from any $k$ column combinations of $\mathcal{C}^a$ using BP decoding and as $b \rightarrow \infty$ we have $b^\prime \rightarrow b$. Note that the raw source data need not be in standard $b \times k$ form. For any positive integer $g$ satisfying $b|g$ and $k|g$, the generator matrix $G_{\mathcal{C}^a}$ should work fine for different arrangements of the data block matrix such as $b/g \times kg$. We finally note that the code $\mathcal{C}^a$ is not in two dimensional standard rectangle form as in $\mathcal{C}$. However, we introduced another parameter $b^\prime$ to be able to make asymptotically MDS array BP-XOR codes analogous to standard MDS array codes defined over rectangle shape binary matrices. 

For a given fixed code rate $r$ and $n$, let us define $\epsilon(b,n)$ to be the maximum coding overhead\footnote{Since columns of $\mathcal{C}^a$ may have different sizes, the overhead depends on which $k$ columns are used for reconstruction. Also note that the coding overhead also depends on the number of columns $n$ in the code, so called array code blocklength.} of $\mathcal{C}^a$ satisfying $b^\prime = (1 + \epsilon(b,n))b$. The asymptotically optimal overhead property implies that as $\epsilon(b,n) \rightarrow 0$ we have $b\rightarrow \infty$. 

Letting $\sigma$ denote the maximum check node degree of a given array BP-XOR code, we note from \cite{BPXOR} that if $k = \sigma$ it is not hard to show that 
\begin{eqnarray}
  n &\leq& kb + 1 + \max\{k-3,0\}
\end{eqnarray}
the upper bound of which can be arbitrarily large (i.e., for $b \gg 1$) and allow any arbitrarily small $r$ to be possible. However, for $k > \sigma$ it is observed that the array code blocklength $n$ is upper bounded based on a specific choice of $k$ \cite{BPXOR}. In addition, we observe from the same study that for $b \gg 1$ and large enough $k$ i.e., $k > \sigma^2$ we have $n \leq k + \sigma -1$. This also implies that for large enough information block length $k$, the achievable rate will be close to 1, putting a constraint on the code design rate. 

\subsection{Main Result}

We begin with providing the following theorem that sets the necessary condition/s on the parameters for the existence of asymptotically MDS array BP-XOR codes.

\begin{theorem} \label{Thm21}
Let  $\mathcal{C}^a$ be a $[n,k,t,b, b^\prime]$ asymptotically MDS array BP-XOR code such that the maximum coded node degree satisfies $2 < \sigma < (bk-1)/(b^\prime-1)$. Then, we have
\begin{eqnarray}
n &\leq& k + \sigma - 1 + \\
&& \ \ \ \ \left\lfloor \frac{b(k(\sigma^\prime - \sigma) + (\sigma-1)\sigma^\prime) - (\sigma-1)(3\sigma/2 - 1)}{b(k - \sigma^\prime) + \sigma - 1} \right\rfloor \nonumber
\end{eqnarray}
where $\sigma^\prime = \sigma(1 + \epsilon(b,n))$ and $\epsilon(b,n)$ is the coding overhead.
\end{theorem}

\begin{proof}
Since the code is assumed to be MDS, i.e., able to tolerate $n-k$ column erasures of $\mathcal{C}^a$,  each information symbol $v_s \in I$ must appear in at least $n-k+1$ columns, totaling up to \begin{eqnarray}
kb(n-k+1) \label{eqn1}
\end{eqnarray}
minimum appearances in $\mathcal{C}^a$. On the other hand, belief propagation decoding starts decoding from degree-one encoding symbols. So we need at least $n-k+1$ degree-one symbols in distinct columns of  $\mathcal{C}^a$ (in the worst case of $n-k$ column erasures when each may comprise one degree-one symbol). Similarly, we need at least one degree-two, one degree-three, $\dots$, one degree-$(\sigma-1)$ coding symbols to make sure that BP decoding continues. Although it is possible to have multiple degree-two symbols and continue BP decoding, by this choice we are trying to maximize the appearance of information symbols in $\mathcal{C}^a$. Note that if these symbols happen to be in distinct unerased columns, the bound could be tightened, otherwise the bound might still be loose for instance if $ \sigma > k + 1$ which is not usually typical. The rest of the $b^\prime n - (n-k+\sigma-1)$ can have at most $\sigma$ degree. Thus, $C^a$ can have at most
\begin{eqnarray}
\sigma(b^\prime n- (n-k+\sigma-1)) + n - k + \frac{\sigma(\sigma-1)}{2} \label{eqn2}
\end{eqnarray}
appearances of $kb$ information symbols. So we have the inequality $(\ref{eqn1}) \leq (\ref{eqn2})$. We can rewrite (\ref{eqn2}) in a more compact form as
\begin{eqnarray}
\sigma b^\prime n - (\sigma-1)(n-k+\sigma/2)
\end{eqnarray}

Using equation (\ref{eqn1}), and assuming we have $b(k - \sigma^\prime) + \sigma - 1 > 0$, we can collect all terms that includes $n$ and find an upper bound on $n$ as follows,
\begin{eqnarray}
n &\leq& \left\lfloor \frac{(kb+\sigma - 1)(k-1) - (\sigma - 1)(\sigma/2-1)}{b(k - \sigma^\prime) + \sigma - 1 } \right\rfloor \label{eqn3} \\
&=& k + \sigma - 1 + \\
&& \ \ \ \ \left\lfloor \frac{b(k(\sigma^\prime - \sigma) + (\sigma-1)\sigma^\prime) - (\sigma-1)(3\sigma/2 - 1)}{b(k - \sigma^\prime) + \sigma - 1} \right\rfloor \nonumber
\end{eqnarray}
where $\sigma^\prime = \sigma(1 + \epsilon(b,n))$. 
\end{proof}

Note that if $b \rightarrow \infty$ we will have $\sigma^\prime \rightarrow \sigma$ and hence equation (\ref{eqn3}) becomes identical to equation (2) of \cite{BPXOR} except the term $(\sigma-1)(\sigma/2-1)$. This term is essentially what makes the upper bound  improved (tighter).

There are two cases that are interesting to consider for understanding the asymptotical performance. First, if $b$ tends large we will have $\sigma^\prime \rightarrow \sigma$. Hence,
\begin{eqnarray}
n &\leq& k + \sigma - 1 + \left\lfloor \frac{(\sigma-1)\sigma}{k - \sigma} \right\rfloor - \textbf{1}_{(k-\sigma)|(\sigma-1)\sigma} \nonumber 
\end{eqnarray}
where $\textbf{1}_{A}$ is logical one if $A$ is true, otherwise it is zero. This indicator function is used due to the flooring operation and $\sigma$ only equals to $\sigma^\prime$ in the limit. Thus, if the code becomes array MDS in the limit, there remains no dependence of $n$ on $b$. On the otherhand, if we let large but fixed $b \leq k$, and if $k$ gets large, we shall have
\begin{eqnarray}
n &\leq& k + \sigma^\prime - 1  \nonumber \\
&=& k + \sigma(1 + \epsilon(b,n)) - 1 \label{eqn9}
\end{eqnarray}
which can be made arbitrarily large if we choose $\epsilon(b, n) \rightarrow \infty$ for a fixed $b$ and large $n$. This essentially demonstrates that as the array BP-XOR code becomes near-optimal in terms of recovery performance, the upper bound on the number of code columns $n$ can dramatically be improved. 

Although the desirable properties of the coding overhead are found, we still need specific constructions to quantify or bound the coding overhead and hence present tighter bounds on $n$ (and $r$) for a specific construction. Based on this observation, we shall present a code construction method that uses the result of Theorem \ref{Thm21} and has an appropriate $\epsilon(b,n)$ with the properties as summarized below.

\begin{itemize}
\item For fixed $k$ and rate $r$ (i.e., fixed $n$), as $b \rightarrow \infty$ we have vanishing coding overhead,  $\epsilon(b,n) \rightarrow 0$.
\item For fixed $b$ and rate $r$, as $n \rightarrow \infty$ we have a diverging coding overhead, $\epsilon(b,n) \rightarrow \infty$.
\end{itemize}

\section{Discrete Geometry Constructions of Asymptotically-MDS array BP-XOR codes}
\label{SectionMoj}

In this section, we will introduce a particular construction of asymptotically MDS array BP-XOR codes based on discrete geometry \cite{Moj} and show that they can be regarded as a special type of the class of asymptotically MDS BP-XOR codes.

The discrete geometry construction is known as Mojette codes which are based on discrete version of Radon Transform \cite{Radon}, and can be used to generate redundancy not just for rectangle two dimensional data grid but also for any convex data grid. In  our study, we consider matrix (rectangle) data and let encoder compute a linear set of projections at angles specified by a couple of coprime integers $(p,q)$ from a $b \times k$ discrete data structure $f:(z,l) \rightarrow \mathbb{N}$. Suppose that we generate $n$ projections with parameters $\{(p_i, q_i), 0 \leq i \leq n-1\}$. The length of the projection $i$, denoted by $b_i$, is a function of the number of projections $n$, the angle parameters $(p_i,q_i)$ and the data grid size $b \times k$. It can be expressed in a closed form as follows \cite{Moj},
\begin{eqnarray}
b_i =  |p_i| (k-1) +  |q_i| (b-1) + 1
\end{eqnarray}

Note that in this construction, generated projections can be treated as the columns of the asymptotically-MDS BP-XOR code. An example code with parameters $k = 3$, $b = 4$ with $n = 3$ projections with parameters $(-1,1),(1,0),(1,1)$ is shown in Fig. \ref{fig:BSW}. Each bin or symbol of the $i$-th projection, based on $(p_i,q_i)$, can be computed as given by the following compact formulation
\begin{align}
& M_{(p_i,q_i)}f(m + (b-1)q_iu(q_i) +(k-1)p_iu(p_i))  \\ & \ \ \ \  \ \ \ \ \ \ \ \ \ \ \ \ \ \ \ \ \ \ \ = \bigoplus_{z=0}^{b-1}\bigoplus_{l=0}^{k-1} f(z,l) \delta_{m  + zq_i + lp_i}\label{moj1}
\end{align}

\begin{figure}[t!]
\centering
\includegraphics[angle=0, height=48mm, width=78mm]{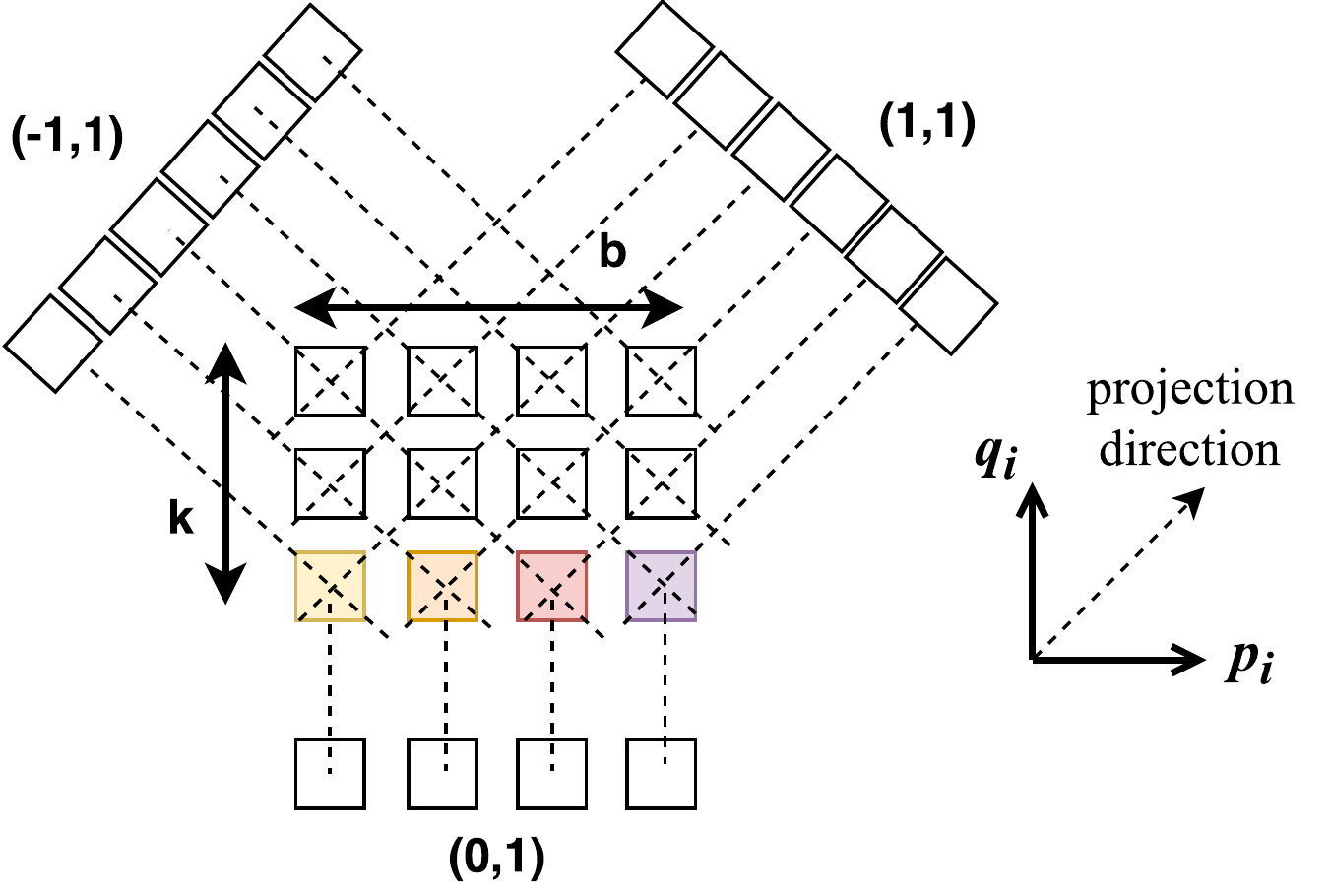}
\caption{A simple illustration of the projection concept and Mojette coding. }\label{fig:BSW}
\end{figure}
for all $m$ values satisfying the inequality,
\begin{gather*}
- (b-1)q_iu(q_i) -(k-1)p_iu(p_i) \nonumber \\
 \leq m \leq  \nonumber
\\ b_i  - (b-1)q_iu(q_i) -(k-1)p_iu(p_i) - 1
\end{gather*}
where $\bigoplus$ stands for Boolean XOR operation, $u(.)$ is the discrete unit function and  $\delta_i$ is Kronecker delta function which are given by
\[
u(s)=
\begin{cases}
1, & \textrm{if }s > 0 \\
0, & \textrm{Otherwise }
\end{cases}
, \ \ \
\delta_i=
\begin{cases}
0, & \textrm{if }i\not=0 \\
1, & \textrm{if }i = 0
\end{cases}
\]

Mojette codes can be decoded using BP algorithm and the exact reconstruction of user data matrix is possible if the projection parameters $(p_i,q_i)$ are selected judiciously according to the following Katz criterion. 

\begin{theorem}
For a given asymptotically-MDS BP-XOR code defined by $n$ projections with parameters $(p_i, q_i)$  on a $b \times k$ data matrix, exact data reconstruction is possible using iterative BP if
\begin{align}
\sum_{i=0}^{n-1} |p_i| \geq b \textrm{ or } \sum_{i=0}^{n-1} |q_i| \geq k
\end{align}
\end{theorem}

\begin{proof}
The proof can be found in \cite{katz}.
\end{proof}

According to Theorem \ref{Thm21}, the maximum degree of the coded symbols play key role in the attainable block length of the BP-XOR codes. Thus, next we find the maximum degree number in the case of Mojette transform codes and see that this parameter can be adjusted based on the selection of projection parameters $(p_i, q_i)$. The following theorem quantifies this number.

\begin{theorem}
Let us use $\sigma_i, i \in \{1,2,\dots,n\} $ to denote the maximum degree of the $i$th projection with parameters $(p_i, q_i)$. We have $\sigma_i = \min\{\lceil b/|p_i| \rceil, \lceil k/|q_i| \rceil \}$ and hence $\sigma = \max_i\{\sigma_i\}$.
\end{theorem}

\begin{proof}
 Considering the equation (\ref{moj1}) and the worst case scenario, we would like to find the number of $l$ and $z$ values such that $zq_i + lp_i = -m$. It is not hard to see that the maximum number of $z$ values that can satisfy this equation is given by $\lceil k/|q_i| \rceil$ due to  $0 \leq z \leq k-1$. Similarly, the maximum number of $l$ values  that can satisfy this equation is given by $\lceil b/|p_i| \rceil$ due to  $0 \leq l \leq b-1$. Since the number of possibilities for $z$ and $l$ are also constrained by the two dimensional rectangular shape, we have the maximum encoding symbol degree equal to the minimum of the two i.e., $\sigma_i = \min\{\lceil b/|p_i| \rceil, \lceil k/|q_i| \rceil \}$. Thus, the maximum degree of all the code symbols is given by the maximum degree of all the projections i.e., $\sigma = \max_i \{ \min\{\lceil b/|p_i| \rceil, \lceil k/|q_i| \rceil \} \}$. 
\end{proof}

Next, we quantify the coding overhead for Mojette transform based asymptotically MDS BP-XOR codes by considering $k=\sigma$ and $k > \sigma$ cases separately. 

\subsection{ Case $k = \sigma$}

First of all, note that depending on the choices of $(p_i, q_i)$, the code overhead as well as the maximum degree of the code can change. Although, there are multiple choices for  $k = \sigma$, we provide the typical choice below that also ensures block length.

\begin{construction} \label{Cons33}
Let us consider the following choice of coprime integers,
\begin{align}
q_i = 1, p_i &\in \mathfrak{T} = \left\{-\left\lfloor\frac{n-1}{2}\right\rfloor,\dots,-1,0,1,2,\dots,\left\lceil\frac{n-1}{2}\right\rceil\right\} \label{option1}
\end{align}
where $\mathfrak{T}$ is known as canonical enumeration of integers \cite{OEIS} that goes with the name \emph{A007306} and satisfies $gcd(p_i, q_i) = 1$ for $i=0,\dots,n-1$. 
\end{construction}

Note that this construction satisfies the Katz criterion simply because collecting any $k$ projections will lead us to have $\sum |q_i| = k$. If we use the coprime integers as given by the Construction \ref{Cons33}, we have $q_i$ never equal to zero and $\sigma_i = \min\{\lceil b/ \lceil (n-1)/2 \rceil,k\}$. We note that we have $\sigma = k$ for  $b \gg 1$. We next quantify the coding overhead for this particular construction and show the asymptotically optimal property.

\begin{theorem}
For the Mojette code with parameters as given in Construction \ref{Cons33}, for $b \gg 1$, we have
\begin{eqnarray}
\epsilon(b,n) \approx \frac{n(2-r)(nr-1)}{4b}
\end{eqnarray}
where $r$ is the fixed rate of the array BP-XOR code.
\end{theorem}

\begin{proof}
See appendix A  the proof of this theorem.
\end{proof}

For fixed $r$ and $k$ (i.e., fixed $n$), if $b \rightarrow \infty$ then it is clear that $\epsilon(b,n) \rightarrow 0$ proving the asymptotical property. On the other hand, for fixed $r$ and $b$, if $n \rightarrow \infty$ then we have $\epsilon(b,n) \rightarrow \infty$. In fact, it is not hard to see that $\epsilon(b,n) = O(n^2)$. Therefore, due to these desirable properties of the overhead and considering the inequality (\ref{eqn9}), we can make $n$ arbitrarily large. Particularly we can find the following lower bound on $n$ for $k = rn = \sigma$ and $r>0.5$,
\begin{eqnarray}
n \leq rn + rn \left(1 + \frac{n(2-r)(nr-1)}{4b} \right) - 1
\end{eqnarray}
which yields the inequality
\begin{eqnarray}
n-2nr \leq \frac{n^3r^2(2-r)}{4b} 
\Rightarrow n \geq \sqrt{\frac{4b(1-2r)}{r^2(2-r)}}
\end{eqnarray}

This final lower bound shows that the value for the block length $n$ can be arbitrarily large for judiciously selected large $b$. Note that the case $k = \sigma$ has the least constraint on the code block length for any MDS array BP-XOR code. The case $k > \sigma$ is more interesting for the class of asymptotically MDS array BP-XOR codes. 

\subsection{ Case $k > \sigma$}

With classical array BP-XOR codes, the block length $n$ is constrained by the following upper bound for $b \gg 1$,
\begin{eqnarray}
n \leq k + \sigma -1 + \left\lfloor \frac{\sigma(\sigma-1)}{k-\sigma} \right\rfloor - \textbf{1}_{(k-\sigma)|(\sigma-1)\sigma}
\end{eqnarray}
which is the same for asymptotically MDS array BP-XOR codes as mentioned in Section II. However, as the block length gets large as well, we shall no longer have constraints on the size of the block length for asyptotically MDS BP-XOR codes. 

Next, we provide another set of parameters for Mojette code that shall satisfy $k > \sigma$. The possibilities of the pair $(p_i, q_i)$ selection for making $k > \sigma$ is not unique. We will consider the typical class as given in construction \ref{cons35}.

\begin{construction} \label{cons35}
Let us consider the following choice of coprime integers for $n$ projections, 

\begin{align}
q_i &= q_e > 0, \nonumber \\
p_i &\in \mathfrak{U} = \left\{\left\lceil-n+1\right\rceil_{odd},\dots,-1,1,3,\dots,\left\lceil n-1 \right\rceil_{odd}\right\} \label{option2}
\end{align}
where $q_e$ is a positive even number, and $\left\lceil . \right\rceil_{odd}$ rounds to the next biggest odd integer of the argument, respectively. 
\end{construction}

Note that using construction \ref{cons35}, it is easy to verify that we have $GCD(p_i, q_i) = 1$. Also, we have $k > \sigma = \max_i \{ \min\{\lceil b/|p_i| \rceil, \lceil k/|q_i| \rceil \} \} = \lceil k/q_e \rceil$. It is of interest to quantify the coding overhead to be able to find the upper bounds on the code block length. 

\begin{theorem}
For the Mojette code with parameters as given in construction 3.5, for $b \gg 1$, we have
\begin{align}
\epsilon(n,b) & \approx \\
& \frac{\lceil k/q_e \rceil}{kb}
\left((k-1) \left(n - \frac{\lceil k/q_e \rceil}{2} \right)  + (b-1)q_e + 1 \right) - 1 \nonumber
\end{align}
where $q_e$ is a positive even number, and $\left\lceil . \right\rceil_{odd}$ rounds to the next biggest odd integer of the argument, respectively.
\end{theorem}

\begin{proof}
See appendix B for the proof of this theorem.
\end{proof}

Note that as long as $q_e | k$, we have $\epsilon \rightarrow 0$ for large $b$ demonstrating the asymptotically optimal overhead property. Similarly, for fixed $r$ and $b$, if $n \rightarrow \infty$ then we have $\epsilon(n,b) \rightarrow \infty$ satisfying the second desirable property. 

Finally, using equation (\ref{eqn9}) we can express the upper bound on $n$ as follows,
\begin{equation}
n \leq k +  \frac{\sigma \lceil k/q_e \rceil}{kb}
\left((k-1) \left(n - \frac{\lceil k/q_e \rceil}{2} \right)  + (b-1)q_e + 1 \right) - 1
\end{equation}

Since it is hard to see that with this result we improve the upper bounds on the code block length, in the next section, we provide some numerical results that compute the upper bounds for comparison.

\section{Numerical Results}
\label{SectionSim}

\begin{figure}
  \centering
  \includegraphics[width=0.5\textwidth]{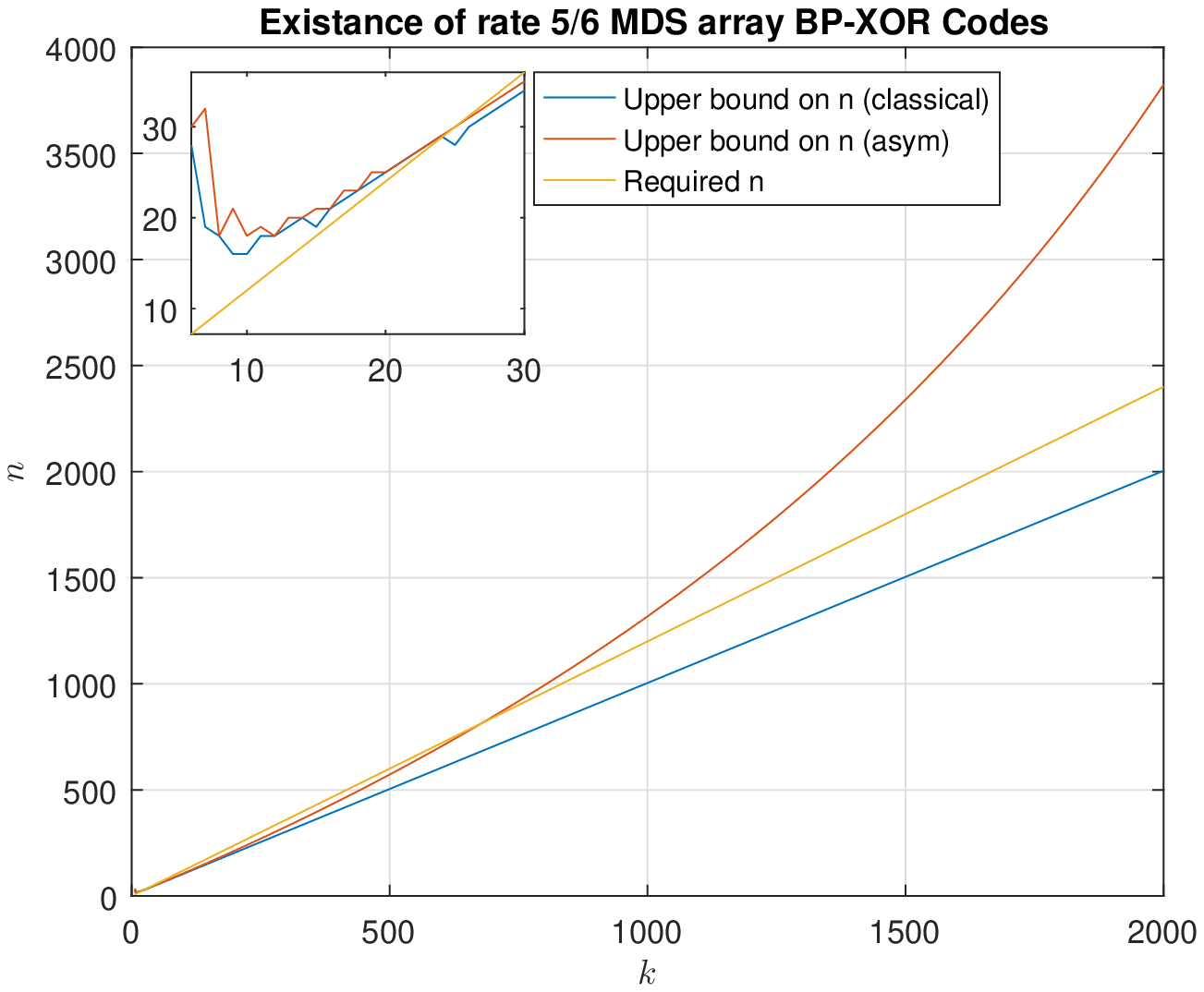}
  \caption{Upper bounds on $n$ as a function of $k$ for $b=10000$.}
\end{figure}

\begin{figure}
  \centering
  \includegraphics[width=0.5\textwidth]{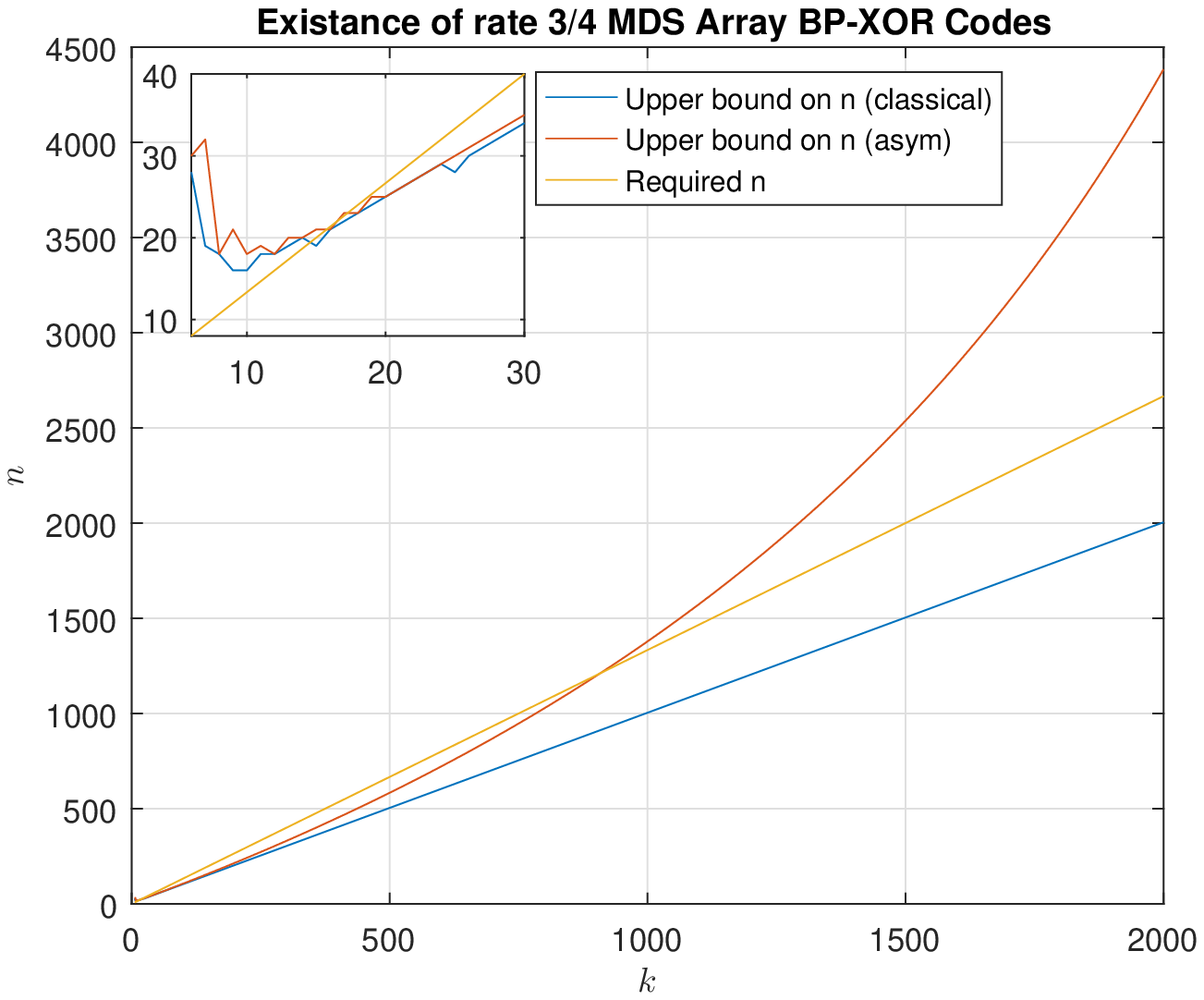}
  \caption{Upper bounds on $n$ as a function of $k$ for $b=10000$.}
\end{figure}

\begin{figure}
  \centering
  \includegraphics[width=0.5\textwidth]{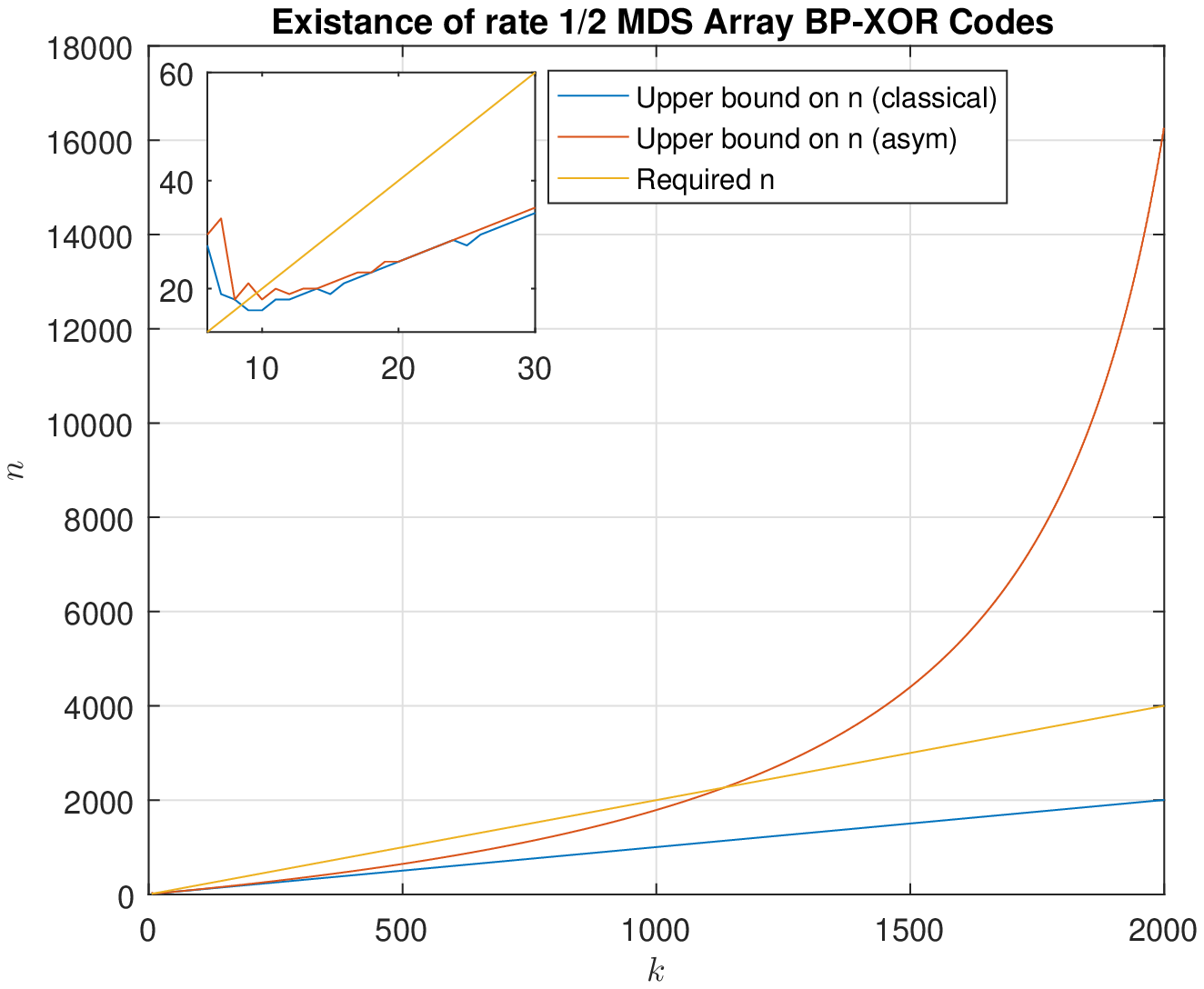}
  \caption{Upper bounds on $n$ as a function of $k$ for $b=10000$.}
\end{figure}

\begin{figure}
  \centering
  \includegraphics[width=0.5\textwidth]{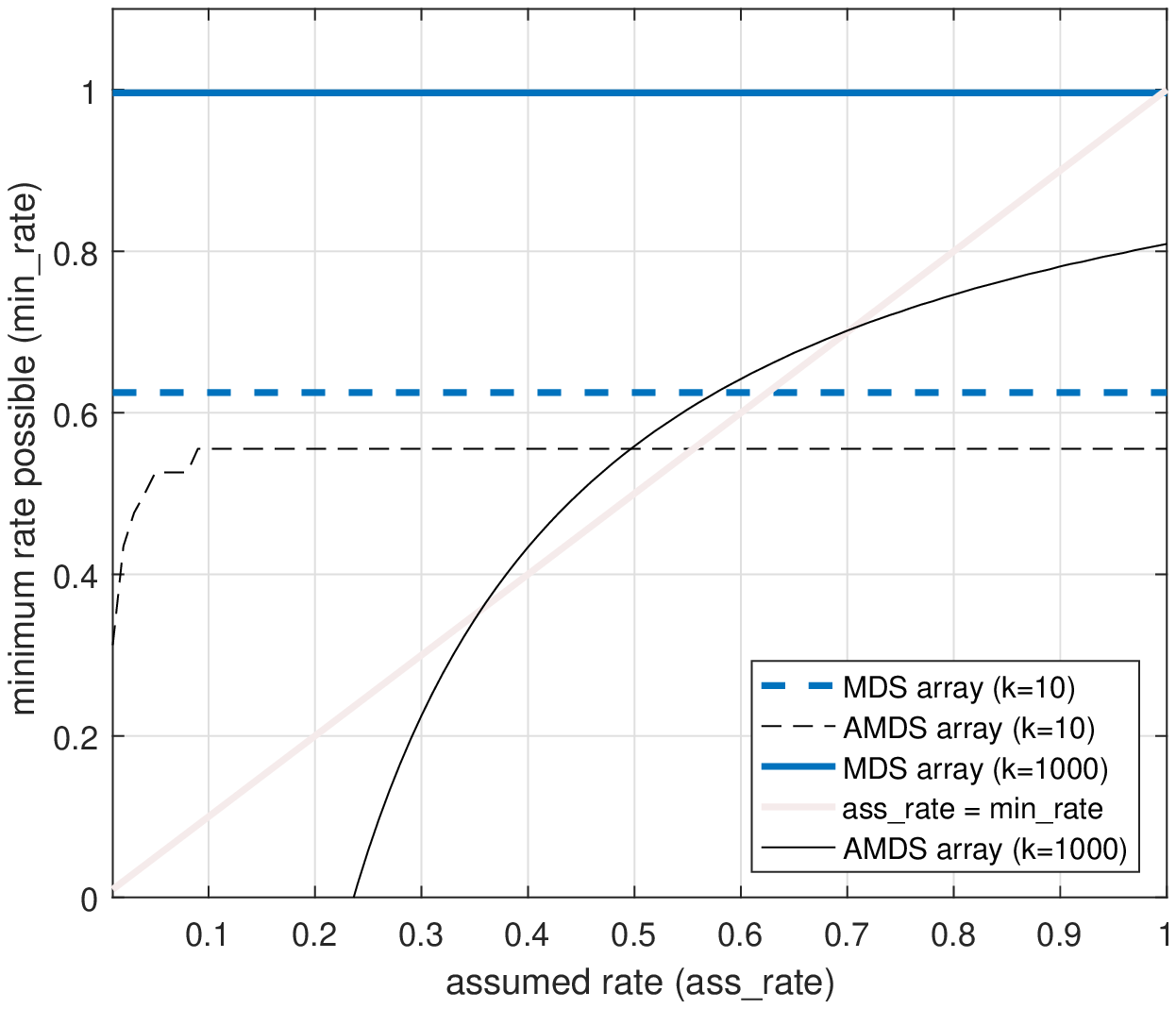}
  \caption{Upper bounds on $n$ as a function of $k$ for $b=10000$.}
\end{figure}

Let us consider $q_e = 2$ and a large $b$ value, such as $b=10000$ (this choice is completely arbitrary) and compare the upper bounds on $n$ with using classical MDS array BP-XOR codes and their asymptotically optimal version proposed in our study, abbreviated as AMDS. We present results in Fig.1, Fig. 2 and Fig. 3 each corresponding to three different rates $5/6, 3/4, 1/2$, respectively as example use cases. These results demonstrate that as the code rate decreases, classical MDS array BP-XOR codes are only possible for very small values of $k$. On the other hand, although the same is true for asymptotically MDS BP-XOR codes for small $k$, it is also observed that for large enough $k$ our bounds are bigger than the required $n$ (fixed by the code rate), allowing possible constructions to achieve the corresponding rate asymptotically MDS array BP-XOR code such as Mojette construction we have provided in previous sections. These figures also present the upper bound behavior for small $k$ on the left corner of each plot. Plots include a curve ``Required $n$" to denote the required value for $n$ for the corresponding rate $r=k/n$ code. 

In order to see clearly the range of rates that are possible with both constructions, Fig. 5 depicts the minimum rate that is possible as a function of the assumed rate. Note that with asymtotically MDS array BP-XOR codes, the upper bound on $n$ depends on the coding overhead which is a function of rate of the code. Thus, the minimum code rate changes as the assumed code rate changes. For each assumed rate, we calculate the upper bound and then compute the minimum code rate possible. With respect to classical MDS BP-XOR codes, since the upper bound does not change with varying assumed rate (since the coding overhead is always zero), the curves turns out to be flat. 

According to Fig. 5, the region that lies above the curves are the possibilities of the code rate. However, there is no guarantee each and every assumed rate would be achievable. However, as can be seen as $k$ gets large it becomes impossible to construct classical MDS array BP-XOR codes with rate smaller than 1. In contrast, by relaxing the exact MDS condition (such as adapting asymptotically MDS constructions), we can improve the the region of possibilities  for better achievability. With this study, we have just provided one simple construction based on discrete geometry (with judicious selection of parameters) that helps improve the upper bounds on the code block length $n$. Other constructions may help improve the results presented in this subsection.

\section{Conclusion}
\label{SectionConc}
Array BP-XOR codes are attractive data protection schemes for low-complexity and optimal reliability. Their finite versions are shown to have limitations on the maximum block length when the coding symbol degree is particularly lower than the data size. We have shown in this study, this limitation can greatly be relaxed by extending the original optimal class to asymptotically optimal class. We have also have shown one particular code construction based on discrete geometry that satisfies all the requirements of being asymptotically MDS array BP-XOR codes. These codes can be encoded and decoded in linear time with the block length and the achievable bound on the block length is far from that of the finite counterpart. 

\appendices
\section{Proof of Theorem 2.4}
Let us start by defining the following utility function,
\begin{eqnarray}
\varphi(x) = \left\lfloor\frac{x}{2}\right\rfloor\left(\left\lfloor\frac{x}{2}\right\rfloor + 1\right) \textmd{ for $x$ } \geq 0.
\end{eqnarray}
Also let $I_t = \{0,1,\dots,t-1\}$. Using these definitions, we state the following lemma next.

\emph{Lemma A.1:} For the projection set given as in (\ref{option1}), we have the sum $\sum_{i=0}^{t-1} |p_i|$ that can be expressed in a closed form using the utility function
\[
\sum_{i \in I_t} |p_i| = \frac{1}{2} \left( \varphi(t) + \varphi(t-1) \right) =
\begin{cases}
\frac{t^2-1}{4}, & \textrm{if $t$ is odd} \\
\frac{t^2}{4}, & \textrm{if $t$ is even}
\end{cases}
\]
This lemma can easily be proved by considering $t$ odd and even cases using induction, separately. Note that the integer sequence $\sum_{i \in I_t} |p_i|$ is given by \emph{A002620} \cite{OEIS}. Using this result, for a given pair of projections $t_2$ and $t_1$ satisfying $t_2 > t_1$, with the associated projection parameters $(p_i^{(t_2)}, q_i^{(t_2)} = 1)$ and $(p_i^{(t_1)}, q_i^{(t_1)} = 1)$ selected based on construction 2.2 (\ref{option1}), we can deduce that
\begin{eqnarray}
\frac{t_2^2 - t_1^2 - 1}{4} \leq \sum_{i=0}^{t_2-1} |p_i^{(t_2)}| - \sum_{j=0}^{t_1-1} |p_j^{(t_1)}| \leq \frac{t_2^2 - t_1^2 + 1}{4} \label{eqnsum}
\end{eqnarray}

Note that since $q_i=1$, it is sufficient to collect $k$ projections for perfect reconstruction. Thus, the upper/lower bounds given in equation (\ref{eqnsum}) are particularly useful if we set $t_2 = n$ and $t_1 = n - k$ to be able find the contributions from the largest $k$  projections in the sum that appears in the worst case coding overhead expression. Let $i^\prime$ be the index such that $p_{i^\prime}^{(t_2)} = p_0^{(t_1)}$ and define the set
\begin{eqnarray}
S = \{i^\prime, i^\prime + 1, \dots, i^\prime + n - k - 1\}
\end{eqnarray}

The worst case coding overhead in this case is given by the following
\begin{align}
\epsilon(n,b) &= \frac{1}{kb} \left( \sum_{i \in I_t \backslash S} |p_i| (k-1) + |q_i| (b-1) + k \right) - 1 \\
&= \frac{(b-1)k + \frac{k-1}{2}\left(\varphi(n) + \varphi(n-1)\right)}{kb}  \nonumber \\
& \ \ \ \ \ + \frac{-\frac{k-1}{2}\left( \varphi(n-k) + \varphi(n-k-1)\right) + k}{kb}  - 1 \label{eqn10} \\
&= \frac{k-1}{2kb}\left(\varphi(n) + \varphi(n-1) - \varphi(n-k) - \varphi(n-k-1)\right)
\end{align}
where Equation (\ref{eqn10}) follows from the conjecture \emph{Lamma 1}.
Again, using conjecture \emph{Lamma 1} and Equation (\ref{eqnsum}), and through some algebra, we can bound the worst case coding overhead as follows,
\begin{eqnarray}
\frac{k-1}{4kb}\left(2kn-k^2 - 1\right)
\leq \epsilon(n,b) \leq \frac{k-1}{4kb}\left(2kn-k^2 + 1\right)
\end{eqnarray}
which can  be accurately approximated for $b \gg 1$ as
\begin{eqnarray}
\epsilon(n,b) \approx \frac{k-1}{4kb}\left(2kn-k^2\right)  = \frac{k-1}{4b}(2n-k)\label{eqn13}
\end{eqnarray}
from which the result follows.

\section{Proof of Thm.}

Let us start by stating the following lemma.

\emph{Lemma B.1:} For the projection set given as in (\ref{option2}) with $t$ projections, we have the sum $\sum_{i=0}^{t-1} |p_i|$ that can be expressed in a closed form using the utility function
\[
\sum_{i=0}^{t-1} |p_i| =
\begin{cases}
\frac{t^2+1}{2}, & \textrm{if $t$ is odd} \\
\frac{t^2}{2}, & \textrm{if $t$ is even}
\end{cases}
\]

\emph{Proof:} Let us consider the sum for even and odd $t$ separately. First we assume $t$ to be odd. Let us define the set
\begin{align}
\mathfrak{U}_{a} = \left\{\left\lceil-t+1\right\rceil_{odd}-a,\dots,-1-a,1-a,\dots,\left\lceil t-1 \right\rceil_{odd}-a\right\}
\end{align}
and notice that $\mathfrak{T} = \mathfrak{U}_0 \cup \mathfrak{U}_{1}$. Since these sets are disjoint, we have
\begin{eqnarray}
\sum_{i \in \mathfrak{T}} |p_i| = \sum_{i \in \mathfrak{U}_0} |p_i| + \sum_{i \in \mathfrak{U}_{1}} |p_i| = 2\sum_{i \in \mathfrak{U}_{1}} |p_i| + 1
\end{eqnarray}
Using this relationship and the result of Lemma A.1, we can express
\begin{eqnarray}
\sum_{i \in \mathfrak{U}_0} |p_i| &=& \sum_{i \in \mathfrak{U}_1} |p_i| + 1 \\
&=& \frac{\sum_{i \in \mathfrak{T}} |p_i| - 1}{2} + 1 \\
&=& \frac{\frac{1}{2}(\phi(2t)-\phi(2t-1))-1}{2} + 1 = \frac{t^2+1}{2}.
\end{eqnarray}

Now let us assume $t$ to be even. For this particular assumption we can rewrite
\begin{eqnarray}
\mathfrak{T} = \mathfrak{U}_0 \cup \mathfrak{U}_{1} \cup \{t\}
\end{eqnarray}
Using this observation and the result of Lemma A.1, we can express
\begin{eqnarray}
\sum_{i \in \mathfrak{U}_0} |p_i| &=& \sum_{i \in \mathfrak{U}_1} |p_i| = \frac{\sum_{i \in \mathfrak{T}} |p_i| - t}{2} \\
&=& \frac{1}{2} \left( \frac{(2t+1)^2 - 1}{4} - t \right) = \frac{t^2}{2}
\end{eqnarray}
which completes the proof of the lemma.

According to Theorem 3.1, we need to have $\sum_{i=0}^{t-1}|q_i| = tq_e \geq k$. This implies $t = \lceil k/q_e \rceil$ projections are sufficient for perfect reconstruction. For a given pair of projections $t_2$ and $t_1$ satisfying $t_2 > t_1$, with the associated projection parameters $(p_i^{(t_2)}, q_i^{(t_2)} = q_e)$ and $(p_i^{(t_1)}, q_i^{(t_1)} = q_e)$ selected based on construction 3.6, we can deduce that
\begin{eqnarray}
\frac{t_2^2 - t_1^2 - 1}{2} \leq \sum_{i=0}^{t_2-1} |p_i^{(t_2)}| - \sum_{j=0}^{t_1-1} |p_j^{(t_1)}| \leq \frac{t_2^2 - t_1^2 + 1}{2} \label{eqnsumX}
\end{eqnarray}

To be able find the contributions from the largest $\lceil k/q_e \rceil$  projections, we set $t_2 = n$ and $t_1 = n - \lceil k/q_e \rceil$. Using similar arguments to previous appendix, we can express the worst case coding overhead in this case as follows
\begin{align}
\epsilon(n,b) &= \frac{1}{kb} \left( \sum_{i \in I_t \backslash S} |p_i| (k-1) + |q_i| (b-1) + \lceil k/q_e \rceil \right) - 1 \\
&= \frac{(b-1)q_e \lceil k/q_e \rceil + \frac{k-1}{2}\left(\varphi(n) + \varphi(n-1)\right)}{kb}  \nonumber \\
& \ \ \ \ \ - \frac{\frac{k-1}{2}\left( \varphi(n-\lceil k/q_e \rceil) + \varphi(n-\lceil k/q_e \rceil-1)\right)}{kb}  \label{eqn10} \\
& \ \ \ \ \ + \frac{\lceil k/q_e \rceil}{kb} - 1
\end{align}

Using equation (\ref{eqnsumX}) and $b \gg 1$, we can accurately approximate the worst case coding overhead as,
\begin{align}
\epsilon(n,b) & \approx \\
& \frac{\lceil k/q_e \rceil}{kb}
\left((k-1) \left(n - \frac{\lceil k/q_e \rceil}{2} \right)  + (b-1)q_e + 1 \right) - 1 \nonumber
\end{align}

\end{document}